\documentclass[conference]{IEEEtran}
\usepackage[colorlinks,
linkcolor=red,
anchorcolor=green,
citecolor=blue
]{hyperref} 
\usepackage{amsmath,amssymb}
\usepackage{latexsym}
\usepackage{CJK}
\usepackage{cite}

\usepackage{bm}
\usepackage{color, soul}

\usepackage{graphicx}
\usepackage{epstopdf}
\usepackage{epsfig}
\usepackage{subfigure} 

\usepackage{verbatim}  

\usepackage{mathrsfs}	
\usepackage{algorithm} 
\usepackage{algorithmic} 
\usepackage{booktabs}
\usepackage{textcomp}
\usepackage{multirow}
\usepackage{lettrine}
\usepackage[scaled=0.92]{helvet}

\usepackage{mathrsfs}
\usepackage{amsmath}
\usepackage{amssymb}

\usepackage{algorithm}
\usepackage{algorithmic}

\usepackage{amsthm} 

\DeclareMathOperator*{\argmine}{arg\, }

\newtheorem{thm}{Theorem}
\newtheorem{cor}{Corollary}
\newtheorem{lem}{Lemma}

\begin{document}
\title{Energy-efficient 3D UAV-BS Placement Versus Mobile Users' Density and Circuit Power}

\author{\IEEEauthorblockN{Jiaxun~Lu, Shuo~Wan, Xuhong~Chen, Pingyi~Fan}

\IEEEauthorblockA{
Tsinghua National Laboratory for Information Science and Technology(TNList),\\
Department of Electronic Engineering, Tsinghua University, Beijing, P.R. China, 100084,\\
Email: \{lujx14, s-wan13, chenxh13\}@mails.tsinghua.edu.cn, fpy@mail.tsinghua.edu.cn}
}

\maketitle


\begin{abstract}

Properly 3D placement of unmanned aerial vehicle mounted base stations (UAV-BSs) can effectively prolong the life-time of the mobile ad hoc network, since UAVs are usually powered by batteries. This paper involves the on-board circuit consumption power and considers the optimal placement that minimizes the UAV-recall-frequency (UAV-RF), which is defined to characterize the life-time of this kind of network. Theoretical results show that the optimal vertical and horizontal dimensions of UAV can be decoupled. That is, the optimal hovering altitude is proportional to the coverage radius of UAVs, and the slope is only determined by environment. Dense scattering environment may greatly enlarge the needed hovering altitude. Also, the optimal coverage radius is achieved when the transmit power equals to on-board circuit power, and hence limiting on-board circuit power can effectively enlarge life-time of system. In addition, our proposed 3D placement method only require the statistics of mobile users’ density and environment parameters, and hence it's a typical on-line method and can be easily implemented. Also, it can be utilized in scenarios with varying users' density.

\end{abstract}


\IEEEpeerreviewmaketitle

\section{Introduction}\label{Sec:Introduction}

Unmanned aerial vehicle mounted base stations (UAV-BSs) have recently gained wide popularity as a feasible solution to provide wireless coverage in a rapid manner. In this system, UAVs are often powered by batteries\cite{gupta2016survey}, and hence the life-time of UAV-BS system is limited by the energy-efficiency. A general strategy to improve energy-efficiency is to adjust the 3D placement of UAVs according to users' density, environment and desired transmit data rate, etc. Moreover, on-board circuit power corresponding to rotors, computational chips and gyroscopes, etc. may also greatly affect energy-efficiency.

There are growing number of works discussing the placement of UAV-BSs subject to coverage range, number of active UAVs and transmit power. In \cite{al2014optimal}, the authors found the optimal hovering altitude that maximize coverage range for single UAV. For scenarios with multiple UAVs, the optimal 3D placement of UAV-BSs was numerically discussed in \cite{bor2016efficient} to maximize the number of covered users and energy-efficiency simultaneously. Also, in \cite{lyu2016placement}, the authors proposed the numerical methods to get the optimal placement and minimum required number of UAVs while getting fully coverage of users. In \cite{Mozaffari2016Efficient}, the optimal hovering altitude and coverage range were analyzed to minimize transmit power. However, in mentioned works, the analysis was only based on the signal to noise ratio (SNR) at the border of coverage range.

Later on, the users inside the coverage range and their density were considered\cite{kalantari2016number,mozaffari2016optimal}. Similarly, the minimum required number of UAVs and optimal 3D placement that minimizes transmit power was proposed in \cite{kalantari2016number} and \cite{mozaffari2016optimal}, respectively. In previous works, the on-board circuit power is not involved, which as previously illustrated, may greatly affect the energy-efficiency of UAV-BS system.

In this paper, we address the importance of on-board circuit power and consider the problem on optimal 3D placement that maximize the life-time of the mobile ad hoc network. We focus on the downlink of UAV-BSs, in which each of the ground users is serviced with fixed data rate. Due to the mobility of users, we assume that the average of users' density is available and invariant in a considered duration. To characterize the life-time of the network, we define the notion of UAV-recall-frequency (UAV-RF), which is the frequency of the active UAVs run out of batteries, and hence maximizing life-time is equivalent to minimizing UAV-RF. To this end, we consider the 3D placement of UAV-BSs separately.

The first is the vertical dimension. By analyzing the coverage scenario with one single UAV, we formulate the problem on finding the optimal hovering altitude that minimizes transmit power for fixed coverage radius. Theoretical results show that the optimal hovering altitude is proportional to the coverage radius, and the slope is only determined by communication environment. That is, in dense scattering environments, the slope is large, and hence UAVs are supposed to fly higher compared with sparse scattering environments.

Apply the derived optimal hovering altitude, we derive the UAV-RF versus environment, coverage parameters and on-board circuit power, where coverage parameters represent the coverage radius, users' density and desired data rate. Analytical results demonstrate that the minimum UAV-RF is achieved when the transmit power equals to the on-board circuit consumption power, and the value of optimal UAV-RF is becoming large in scenarios with dense scattering environment, high on-board circuit consumption power and large users' density and data rate. This indicates that limiting on-board circuit power can effectively prolong the life-time of network. It's worthy to mention that our proposed 3D placement method is a typical on-line method and easily to be implemented, since only the statistics of users' density and environment are needed.

The rest of this paper is organized as follows. In Section \ref{Sec:SysModel}, the system model is introduced. Also, the problem on minimizing UAV-RF is mathematical formulated. Then, the optimal 3D placement of UAVs is discussed in Section \ref{Sec:StaticCoverage}. In Section \ref{Sec:Simulations}, the validity of previous theoretical results and the effectiveness of out proposed optimal 3D placement method are verified by numerical results. Finally, conclusions are given in Section \ref{Sec:Conclusion}.

\section{System Model and Problem Formulation}\label{Sec:SysModel}

This section first models the downlink of UAV-BSs, where users' density is considered with respect to different traffic patterns. Then, the air-to-ground (A2G) channel is provided, and the problem on optimal 3D placement of UAV-BSs is stressed and mathematically formulated.

\subsection{UAV Coverage Model}\label{SubSec:UAVCoverageModel}

Due to the municipal planning of city, there exist multiple subregions in cities, such as entertainment (E), resident (R), transport (T), office (O) and comprehensive (C), and each of them may have unique mobile traffic patterns\cite{xu2016understanding}. Assume the data rate of users is constant, the statistics of users' density can be characterized by the mobile traffic patterns. As shown in Fig. \ref{Fig:UAVCoverage}, a geographical area is divided into several subregions according to their different traffic patterns. The distribution of ground users is modeled by Poison point process (PPP) with density $\lambda_{\rm{u}}\left(\beta,t\right)$, where $\beta=1,2,\cdots,\kappa$ denotes the index of subregions and $t$ is the time index. $\kappa$ is the number of subregions, and in our considered area, $\kappa=5$.

\begin{figure}[htbp]
\centering
\includegraphics[width=0.48\textwidth]{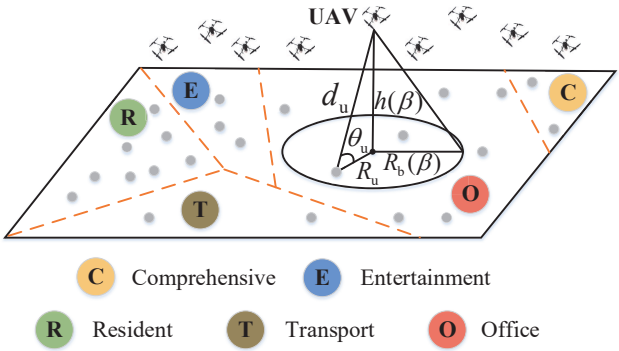}
\caption{UAVs transmitting data to ground users.} \label{Fig:UAVCoverage}
\end{figure}

We focus on the downlink scenario in which UAV-BSs adopt a frequency division multiple access (FDMA) technique to serve each of the ground users with fixed data rate $S_{\rm{u}}$. $S_{\rm{u}}$ is normalized by frequency bandwidth. UAVs assign individual frequency bands to ground users, and hence the frequency interference between UAV-BSs is avoided. Without loss of generality, we assume that the transmit power of each UAV and the available bandwidth are sufficient to meet users' rate requirement.

Since only the average of users' density in each subregion is available, we consider a simple UAV-BS coverage model where UAVs in the same subregion provide the equal coverage radius. This model is called disk-covering model and has been adopted in \cite{mozaffari2016unmanned}. Note that the area of overlaps between adjacent UAV-BS cells is proportional to the area of UAV-BS cells, we can ignore the overlaps and derive the UAV number in each subregion as
\begin{equation}\label{Equ:UAVNumber}
N(\beta) = \frac{A(\beta)}{\pi R^2_{\rm{b}}(\beta)},
\end{equation}
where $R_{\rm{b}}(\beta)$ and $A(\beta)$ denotes the coverage radius and area in each subregion, respectively. Obviously, the coverage radius and UAV number in each subregion are determined by several factors, such as environment, users' density, desired data rate and other practical factors, which will be analyzed in following sections.

\subsection{Air-to-Ground Channel}\label{SubSec:A2GChannel}

The A2G channel model has been analyzed in \cite{al2014modeling,al2014optimal}. The authors showed that the typical A2G channel can be characterized into line-of-sight (LOS) or non-line-of-sight (NLOS) links. Let $\xi=0$ and $\xi=1$ denote the LOS link and NLOS link, respectively. Then, the path loss is given by
\begin{equation}\label{Equ:A2GPathLoss}
L_\xi(R_{\rm{u}},h(\beta)) = \begin{cases}
\left( 4\pi f_{\rm{c}}/c \right)^2 d_u^2\,\eta_{0},& {\xi=0}\\
\left( 4\pi f_{\rm{c}}/c \right)^2 d_u^2\,\eta_{1},& {\xi=1},
\end{cases}
\end{equation}
where $\eta_{0}$ and $\eta_{1}$ are the excessive path loss on the top of the free space path loss (FSPL) for LOS and NLOS links, determined by environment (suburban, urban, dense urban, high-rise urban, or others).  $f_{\rm{c}}$ is the carrier frequency and $c$ is the speed of light. As shown in Fig. \ref{Fig:UAVCoverage}, $d_{\rm{u}} = \sqrt{R_{\rm{u}}^2+h^2(\beta)}$ is the distance between the user of interest and the corresponding UAV. Here, $R_{\rm{u}}$ is the distance between the user of interest and the projection of UAV on ground, and $h(\beta)$ is the hovering altitude of UAV-BSs at the $\beta$-th subregion. Typically, $\eta_{1}\gg\eta_{0}$. That is, the obstacles in propagation paths greatly improve the path loss, and hence higher LOS probability of A2G channel may reduce average of path loss.

According to the results in \cite{al2014modeling} and \cite{al2014optimal}, the LOS probability of A2G channel depends on the environment, such as density and height of buildings, and the elevation angle between user and UAV. The LOS probability can be expressed as\cite{al2014optimal}
\begin{equation}\label{equ:ProbofLOSLink}
P_{0}(R_{\rm{u}}, h(\beta)) = \frac{1}{1 + a \rm{exp}(-b[\theta_{\rm{u}}-a])},
\end{equation}
where $a$ and $b$ are constants determined by environment and $\theta_{\rm{u}}=180/\pi \tan^{-1}(h(\beta)/R_{\rm{u}})$ is the elevation angle\footnote{In \cite{al2014modeling}, the authors provide different LOS probability model. However, different LOS probability model doesn't affect the main conclusions of this paper.}. Then, $P_{1}(R_{\rm{u}}, h(\beta)) = 1-P_{0}(R_{\rm{u}}, h(\beta))$.

In this case, the average path loss of A2G channel can be given by
\begin{equation}\label{Equ:AveragePathLoss}
\begin{split}
&\bar{L}(R_{\rm{u}}, h(\beta)) \\
&= P_{0}(R_{\rm{u}}, h(\beta))L_0(R_{\rm{u}}, h(\beta)) + P_{1}(R_{\rm{u}}, h(\beta))L_1(R_{\rm{u}}, h(\beta))\\
&=\underbrace{(4\pi f_{\rm{c}}/c)^2 d_u^2}_{\rm{FSPL}} \, \underbrace{\big( \eta_1 + P_{0}(R_{\rm{u}}, h(\beta))(\eta_0-\eta_1) \big)}_{\rm{average\,excessive\,path\,loss}}.
\end{split}
\end{equation}
This clearly indicates the individual effects of FSPL and excessive path loss on average path loss. The first part accounts for FSPL, which monotonically increases with $h(\beta)$ due to the growing distance between UAV and user. However, the second part explains the average excessive path loss. Due to high $h(\beta)$ leads to high LOS probability of A2G links, the second part monotonically decreases with $h(\beta)$. 

Hence, for a constrained $R_{\rm{b}}(\beta)$, there exists optimal hovering altitude $h^*(\beta)$ that minimizes the transmit power of UAV. In addition, to minimize the UAV-RF of considered area, the on-board circuit consumption power of UAVs and the 2D arrangement of UAV-BS cells, i.e. $R_{\rm{b}}(\beta)$, also needs to be considered. Next sub-section will formulate the optimization problem on 3D placement of UAVs.

\subsection{Problem Formulation}\label{SubSec:ProblemFormulation}

Take one of the users and UAVs shown in Fig. \ref{Fig:UAVCoverage} for example. Let the allocated transmit power be $P_{\rm{u}}(\xi)$. Then, the channel capacity can be expressed as
\begin{equation*}
S_{\rm{u}} = \log_2\left( 1 + \frac{P_{\rm{u}}(\xi)}{L_\xi(R_{\rm{u}},h(\beta)) N_0} \right),
\end{equation*}
where $N_0$ is the noise power. The transmit power related to the user of interest is
\begin{equation*}
P_{\rm{u}}(\xi,R_{\rm{u}}) = L_\xi(R_{\rm{u}},h(\beta)) N_0 \left( 2^{S_{\rm{u}}} - 1\right),
\end{equation*}
and hence the average transmit power allocated to the user of interest located at $R_{\rm{u}}$ is given by
\begin{equation*}
\bar{P}_{\rm{u}}(R_{\rm{u}}, h(\beta), S_{\rm{u}}) = \bar{L}(R_{\rm{u}}, h(\beta)) N_0 \left( 2^{S_{\rm{u}}} - 1\right).
\end{equation*}

Then, the expectation of transmit power is the integral of all the users inside the coverage of UAV. That is,
\begin{equation}\label{equ:CommunicationPower}
\begin{split}
&P_{\rm{t}}(R_{\rm{b}}(\beta), \lambda_{\rm{u}}(\beta,t), h(\beta), S_{\rm{u}})\\
&= \lambda_{\rm{u}}(\beta,t)\int_0^{R_{\rm{b}}(\beta)} 2\pi R_{\rm{u}} \bar{P}_{\rm{u}}(R_{\rm{u}}, h(\beta), S_{\rm{u}}) \,dR_{\rm{u}}.
\end{split}
\end{equation}

Let the on-board circuit consumption power and the battery capacity of one single UAV-BS be $P_{\rm{c}}$ and $\rm{E}_{\rm{b}}$. Then, the total consumed power and fly time of one of the UAV-BS at the $\beta$-th subregion is $P_{\rm{s}}(\beta) = P_{\rm{t}}(\beta) + P_{\rm{c}}$ and
\begin{equation*}
T_{\rm{h}}(\beta) = {\rm{E}_{\rm{b}}}/{P_{\rm{s}}(\beta)},
\end{equation*}
respectively. Hence, the corresponding UAV-RF can be expressed as 
\begin{equation}\label{Equ:ChangingBatteryNum}
\Phi(\beta) = \frac{N(\beta)}{T_{\rm{h}}(\beta)} = \frac{N(\beta)P_{\rm{s}}(\beta)}{\rm{E}_{\rm{b}}}.
\end{equation}

From Eq. \eqref{Equ:ChangingBatteryNum}, it can be observed that UAV-RF is determined by the UAV number and total consumed power of one single UAV. Describe the 3D placement of the $\beta$-th subregion and the considered area by a two-tuple as ($R_{\rm{b}}(\beta)$,$h(\beta)$) and ($\bm{R}_{\rm{b}}$,$\bm{h}$), where $\bm{R}_{\rm{b}}=(R_{\rm{b}}(1),R_{\rm{b}}(2),\cdots,R_{\rm{b}}(\kappa))$ is the vector representing the coverage radiuses of $\kappa$ subregions. Similarly, $\bm{h}=(h(1),h(2),\cdots,h(\kappa))$ denotes the hovering altitudes. As illustrated in Section \ref{SubSec:A2GChannel}, to minimize the UAV-RF of considered area, the altitudes and coverage radius of UAVs need to be jointly considered. That is,
\begin{flalign}\label{Equ:MinChangingBatteryNumber}
\mathbf{P}_1: \,\,\min_{(\bm{R}_{\rm{b}},\bm{h})} \,\,\, &\sum_{\beta=1}^{\kappa} \Phi(\beta)\\
\textrm{s.t.}\,\,\,&h(\beta)\geq0,R_{\rm{b}}(\beta) > 0.\label{equ:MainOptimizeProblemConstraints}
\end{flalign}
The inequalities shown in \eqref{equ:MainOptimizeProblemConstraints} is the natural constraints of hovering altitude and coverage radius of UAVs. Note that $\kappa$ can be dynamically adjusted. Hence the solution to $\mathbf{P}_1$ can provide optimal 3D placement strategy for arbitrary areas.

\section{3D Placement of UAV-BSs}\label{Sec:StaticCoverage}

In this section, we shall first analyze the optimal hovering altitude of UAV-BSs given the coverage radius. Then, we shall discuss the optimal 2D placement to minimize UAV-RF.

\subsection{Optimal Hovering Altitude}\label{Sec:Coverage4SingleUAV}

The problem on finding the optimal hovering altitude for constrained coverage radius can be expressed as
\begin{flalign}\label{Equ:OptimalAltitude}
\mathbf{P}_{1\text{-}\rm{A}}: \,\,\min_{\bm{h}} \,\,\, &P_{\rm{t}}(R_{\rm{b}}(\beta), \lambda_{\rm{u}}(\beta,t), h(\beta))\\
\textrm{s.t.}\,\,\,&h(\beta)\geq0.\label{equ:OptimalAltitudeConstraints}
\end{flalign}
The inequality shown in \eqref{equ:OptimalAltitudeConstraints} is the natural constraint of altitude. Before solving $\mathbf{P}_{1\text{-}\rm{A}}$, we have the following lemma.
\begin{lem}\label{Lem:ScaleOptimalAltitude}
The transmit power of UAV can be expressed as
\begin{equation}\label{Equ:ScaleTransmitPower}
\begin{split}
P_{\rm{t}}(R_{\rm{b}}(\beta), \lambda_{\rm{u}}(\beta,t), h(\beta), S_{\rm{u}}) =\psi(\beta,t) \Gamma\left(\frac{h(\beta)}{R_{\rm{b}}(\beta)}\right),
\end{split}
\end{equation}
where $\psi(\beta,t)=\lambda_{\rm{u}}(\beta,t) R^4_{\rm{b}}(\beta) \left( 2^{S_{\rm{u}}} - 1\right)$ is the scale factor and $\Gamma(h(\beta)/R_{\rm{b}}(\beta))=P_{\rm{t}}(1, 1, h(\beta)/R_{\rm{b}}(\beta), 1)$ is the kernel function of transmit power.
\end{lem}
\begin{proof}
Substitute the parameters shown in the transmit power function $P_{\rm{t}}(1, 1, h(\beta)/R_{\rm{b}}(\beta), 1)$ into Eq. \eqref{equ:CommunicationPower}, the kernel function of transmit power can be expressed as
\begin{equation}\label{Equ:KernelFunction}
\Gamma\left(\frac{h(\beta)}{R_{\rm{b}}(\beta)}\right) = \int_{0}^{1} 2\pi R_{\rm{u}} \bar{L}(R_{\rm{u}},h(\beta)/R_{\rm{b}}(\beta))N_0 d R_{\rm{u}}.
\end{equation}
In addition, the average path loss corresponding to ($R_{\rm{b}}(\beta)$,$h(\beta)$) can be expressed as 
\begin{equation*}
\bar{L}(R_{\rm{u}}, h(\beta)) = R^2_{\rm{b}}(\beta)\,\bar{L}(R_{\rm{u}}/R_{\rm{b}}(\beta),h(\beta)/R_{\rm{b}}(\beta)),
\end{equation*}
which can be easily derived from Eq. \eqref{Equ:AveragePathLoss}. Thus, 
\begin{equation*}
\begin{split}
&P_{\rm{t}}(R_{\rm{b}}(\beta), \lambda_{\rm{u}}(\beta,t), h(\beta), S_{\rm{u}}) \\
&= \lambda_{\rm{u}}(\beta,t)\int_0^{R_{\rm{b}}(\beta)} 2\pi R_{\rm{u}} \bar{P}_{\rm{u}}(R_{\rm{u}}, h(\beta), S_{\rm{u}}) \,dR_{\rm{u}}\\
&=\lambda_{\rm{u}}(\beta,t)R^2_{\rm{b}}(\beta)\left( 2^{S_{\rm{u}}} - 1\right) \cdot\\
&\quad\quad\quad\quad\int_0^{R_{\rm{b}}(\beta)} 2\pi R_{\rm{u}}\,\bar{L}(R_{\rm{u}}/R_{\rm{b}}(\beta),h(\beta)/R_{\rm{b}}(\beta)) d R_{\rm{u}}\\
&=\lambda_{\rm{u}}(\beta,t)R^4_{\rm{b}}(\beta)\left( 2^{S_{\rm{u}}} - 1\right)\int_0^{1} 2\pi R_{\rm{u}}\,\bar{L}(R_{\rm{u}},h(\beta)/R_{\rm{b}}(\beta)) d R_{\rm{u}}\\
&=\lambda_{\rm{u}}(\beta,t)R^4_{\rm{b}}(\beta)\left( 2^{S_{\rm{u}}} - 1\right) P_{\rm{t}}(1, 1, h(\beta)/R_{\rm{b}}(\beta), 1).
\end{split}
\end{equation*}
This completes the proof.
\end{proof}
The derived results shown in Lemma \ref{Lem:ScaleOptimalAltitude} is quite informative. This clearly indicates the individual effects of the environment and the coverage parameters on the transmit power. That is, $\Gamma(h(\beta)/R_{\rm{b}}(\beta))$ accounts for the environmental statistics, while $\psi(\beta,t)$ explains the scale effects of coverage parameters. Hence, following corollary can be easily derived.
\begin{cor}\label{Cor:SameElevationAngle}
The optimal hovering altitude for fixed $R_{\rm{b}}(\beta)$ can be expressed as 
\begin{equation*}
h^*(\beta)=R_{\rm{b}}(\beta)h_{\rm{n}}^*(\beta),
\end{equation*}
where $h_{\rm{n}}^*(\beta)$ is the optimal hovering altitude that minimizes $\Gamma(h_{\rm{n}}(\beta))$ and is only determined by environment.
\end{cor}
\begin{proof}
According to Lemma \ref{Lem:ScaleOptimalAltitude}, 
\begin{equation*}
\frac{\partial P_{\rm{t}}(R_{\rm{b}}(\beta), \lambda_{\rm{u}}(\beta,t), h(\beta))}{\partial h(\beta)} = 0
\Leftrightarrow
\frac{\partial \Gamma(h(\beta)/R_{\rm{b}}(\beta))}{\partial h(\beta)} = 0.
\end{equation*}
Meanwhile, notice that $\Gamma(h(\beta)/R_{\rm{b}}(\beta))$ only accounts for the environmental statistics, corollary \ref{Cor:SameElevationAngle} can be easily proved.
\end{proof}
Therefore, solving $\mathbf{P}_{1\text{-}\rm{A}}$ is equivalent to finding the optimal hovering altitude that minimize $\Gamma(h_{\rm{n}}(\beta))$. That is,
\begin{equation}\label{Equ:OptimalAltitude}
h_{\rm{n}}^*(\beta) = \argmine_{h_{\rm{n}}(\beta)} \left\{\frac{\partial \Gamma(h_{\rm{n}}(\beta))}{\partial h_{\rm{n}}(\beta)} = 0\right\}.
\end{equation}
From Eq. \eqref{Equ:KernelFunction}, it can be easily derived that
\begin{equation}\label{Equ:PartialEquation}
\begin{split}
&\frac{\partial \Gamma(h_{\rm{n}}(\beta))}{\partial h_{\rm{n}}(\beta)} = 0\\
&\Leftrightarrow \int_0^1 \Big\{ 2h_{\rm{n}}(\beta)\big( \eta_1 + P_{0}(R_{\rm{u}}, h_{\rm{n}}(\beta))(\eta_0-\eta_1) \big) + \\
& \big(R^2_{\rm{u}} + h_{\rm{n}}^2(\beta)\big) \big( \eta_1 + \frac{\partial P_{0}(R_{\rm{u}}, h_{\rm{n}}(\beta))}{\partial h_{\rm{n}}(\beta)}(\eta_0-\eta_1) \big) \Big\} R_{\rm{u}} d R_{\rm{u}} = 0,
\end{split}
\end{equation}
where
\begin{equation}\label{Equ:PartialPLOS}
\frac{\partial P_{0}(R_{\rm{u}}, h_{\rm{n}}(\beta))}{\partial h_{\rm{n}}(\beta)} = \frac{180 b R_{\rm{u}} P_{0}(R_{\rm{u}}, h_{\rm{n}}(\beta))}{\pi (R^2_{\rm{u}} + h_{\rm{n}}^2(\beta))}(1-P_{0}(R_{\rm{u}}, h_{\rm{n}}(\beta))).
\end{equation}
Substitute \eqref{Equ:PartialEquation} and \eqref{Equ:PartialPLOS} into \eqref{Equ:OptimalAltitude}, one can observe that it's hard to derive explicit $h_{\rm{n}}^*(\beta)$. Hence, we design a numerical algorithm to calculate the optimal hovering altitude, as shown in Algorithm \ref{Alg:OptimalAltitude}.
\begin{algorithm}[htbp]
    \caption{Optimal Hovering Altitude}\label{Alg:OptimalAltitude}
    \begin{algorithmic}[1]
    	\item \textbf{Initialize} ~~\\          
        Environmental parameters: $\eta_1$, $\eta_2$, $a$ and $b$;\\
        Input coverage parameters: $R_{\rm{b}}(\beta)$, $\lambda_{\rm{u}}(\beta,t)$, $S_{\rm{u}}$;\\
        Initialize iteration parameters: $h_{\rm{n}}^*(\beta)=h_{\rm{min}}(\beta)=0$, $h_{\rm{max}}(\beta)=1$;\\
        Set the precision $\epsilon=10^{-3}$.\\

        \WHILE {$\left(\frac{\partial \Gamma(h(\beta))}{\partial h(\beta)}|_{h(\beta)=h_{\rm{max}}(\beta)}\frac{\partial \Gamma(h(\beta))}{\partial h(\beta)}|_{h(\beta)=h_{\rm{min}}(\beta)}\right)\geq0$}\label{code:FindMaxh}
        	\STATE $h_{\rm{max}}(\beta) = 10h_{\rm{max}}(\beta)$;\\
        \ENDWHILE \\

        \WHILE {$\frac{\partial \Gamma(h(\beta))}{\partial h(\beta)}|_{h(\beta)=h_{\rm{n}}^*(\beta)} \geq \epsilon$}\label{code:FindAccurateEta}
        	\STATE $h_{\rm{n}}^*(\beta)= \left( h_{\rm{max}}(\beta) + h_{\rm{min}}(\beta)\right)/2$;\\
            \IF  {$\frac{\partial \Gamma(h(\beta))}{\partial h(\beta)}|_{h(\beta)=h_{\rm{n}}^*(\beta)}\geq0$}
            \STATE $h_{\rm{max}}(\beta) = h_{\rm{n}}^*(\beta)$;
            \ELSE
            \STATE $h_{\rm{min}}(\beta) = h_{\rm{n}}^*(\beta)$;
            \ENDIF
        \ENDWHILE \\
       \textbf{Output} ~Optimal hovering altitude $h^*(\beta) = h_{\rm{n}}^*(\beta) R_{\rm{b}}(\beta)$.\\          
    \end{algorithmic}
\end{algorithm}

\subsection{3D Placement of UAVs}\label{Sec:Coverage4Area}

According to Corollary \ref{Cor:SameElevationAngle}, the optimal hovering altitudes of UAV-BSs are determined by the coverage of UAV-BS cells. Hence, $\mathbf{P}_1$ can be rewritten as
\begin{flalign}\label{Equ:Optimal2DArrangement}
\mathbf{P}_{1\text{-}\rm{B}}: \,\,\min_{\bm{R}_{\rm{b}}} \,\,\, &\sum_{\beta=1}^{\kappa} \Phi(\beta)\\
\textrm{s.t.}\,\,\,&h^*(\beta) = R_{\rm{b}}(\beta) h_{\rm{n}}^*(\beta),\tag{\theequation a}\label{equ:Optimize2DArrangementa}\\
&R_{\rm{b}}(\beta) > 0.\tag{\theequation b}\label{equ:Optimize2DArrangementb}
\end{flalign}
The constraint in Eq. \eqref{equ:Optimize2DArrangementa} is the optimal hovering altitude determined by $R_{\rm{b}}(\beta)$, and the constraint in Eq. \eqref{equ:Optimize2DArrangementb} is the natural constraint of coverage radius. As shown in Eq. \eqref{Equ:Optimal2DArrangement}, the UAV-RF of considered area is the sum of the UAV-RF in each subregion. Hence, the optimal problem $\mathbf{P}_{1\text{-}\rm{B}}$ is equivalent to minimizing $\Phi(\beta)$ for all $\beta\in\{1,2,\cdots,\kappa\}$, individually. 

Substitute Eq. \eqref{Equ:UAVNumber} and Eq. \eqref{Equ:ScaleTransmitPower} into Eq. \eqref{Equ:ChangingBatteryNum}, the UAV-RF at the $\beta$-th subregion can be derived as
\begin{equation}\label{Equ:PhiBetaExpression}
\Phi(\beta) = \frac{A(\beta)}{\pi E_{\rm{b}}} \bigg( \lambda_{\rm{u}}(\beta,t)\left( 2^{S_{\rm{u}}} - 1 \right)\Gamma(h_{\rm{n}}^*(\beta))R^2_{\rm{b}}(\beta) + \frac{P_{\rm{c}}}{R^2_{\rm{b}}(\beta)} \bigg).
\end{equation}
With the theoretical results shown in Section \ref{Sec:Coverage4SingleUAV}, the optimal 3D placement of UAVs that minimizes UAV-RF in the considered area can be summarized as follows.
\begin{thm}\label{Thm:Optimal3DPlacement}
The optimal 3D placement of UAV-BSs in considered area is ($\bm{R}_{\rm{b}}$,$\bm{h}$), where the $\beta$-th element of $\bm{R}_{\rm{b}}$ and $\bm{h}$ is
\begin{equation}\label{Equ:Optimal2DArrangement}
R^*_{\rm{b}}(\beta) = \sqrt[4]{\frac{P_{\rm{c}}}{\lambda_{\rm{u}}(\beta,t)\left( 2^{S_{\rm{u}}} - 1 \right)\Gamma(h_{\rm{n}}^*(\beta))}}
\end{equation}
and
\begin{equation}\label{Equ:OptimalHoveringAltitude}
h^*(\beta) = R^*_{\rm{b}}(\beta) h_{\rm{n}}^*(\beta),
\end{equation}
respectively. $h_{\rm{n}}^*(\beta)$ is the optimal hovering altitude corresponding to kernel function $\Gamma(h_{\rm{n}}(\beta))$.
\end{thm}
\begin{proof}
Because the coefficients shown in Eq. \eqref{Equ:PhiBetaExpression} are non-negative, it can be easily derived that
\begin{equation}\label{Equ:InequalityOfPhi}
\Phi(\beta) \geq \frac{2A(\beta)}{\pi E_{\rm{b}}} \sqrt{\lambda_{\rm{u}}(\beta,t)\left( 2^{S_{\rm{u}}} - 1 \right)P_{\rm{c}}\Gamma(h_{\rm{n}}^*(\beta))},
\end{equation}
and the minimum of $\Phi(\beta)$ is achieved when $R_{\rm{b}}(\beta) = R^*_{\rm{b}}(\beta)$, which is shown in Eq. \eqref{Equ:Optimal2DArrangement}. According to Corollary \ref{Cor:SameElevationAngle}, the optimal hovering altitude corresponding to $R^*_{\rm{b}}(\beta)$ is given by Eq. \eqref{Equ:OptimalHoveringAltitude}.
\end{proof}
Note that the optimal coverage radius $R^*_{\rm{b}}(\beta)$ grows with respect to on-board circuit consumption power $P_{\rm{c}}$ and decreases with respect to coverage parameters. With theorem \ref{Thm:Optimal3DPlacement}, the following corollary can be easily derived.
\begin{cor}\label{Cor:TransmitEqualToHovering}
For the optimal 3D placement of UAV-BSs in considered area, the on-board circuit consumption power equals the transmit power. That is,
\begin{equation}\label{Equ:TransmitEqualToHovering}
P_{\rm{c}} = \psi^*(\beta,t) \Gamma(h_0^*(\beta)),
\end{equation}
where $\psi^*(\beta,t)=\lambda_{\rm{u}}(\beta,t) R^{*4}_{\rm{b}}(\beta) \left( 2^{S_{\rm{u}}} - 1\right)$.
\end{cor}
\begin{proof}
Substitute Eq. \eqref{Equ:Optimal2DArrangement} into Eq. \eqref{Equ:ScaleTransmitPower}, Corollary \ref{Cor:TransmitEqualToHovering} can be easily proved.
\end{proof}
Corollary \ref{Cor:TransmitEqualToHovering} is easy to understand by physical meanings. When on-board circuit consumption power is high, large coverage radius can decrease the number of active UAVs. According to Eq. \eqref{Equ:ChangingBatteryNum}, small $N(\beta)$ decreases the effects of high $P_{\rm{c}}$ on UAV-RF, and hence the UAV-RF is reduced. By contrast, when on-board circuit consumption power is low, small coverage radius can decrease transmit power, which also decreases the UAV-RF. In particular, when $P_{\rm{c}}=0$ W, we have $R^{*}_{\rm{b}}(\beta)=0$ m and $h^*(\beta)=0$ m. That is, users can connect to UAV-BSs just at their positions. In this way, since enlarging the number of UAVs doesn't consume more circuit power, the transmit power is saved.

\section{Numerical Results}\label{Sec:Simulations}
In this section, we shall present some numerical results to show the validity of our theoretical results and provide more insights on the effectiveness of proposed optimal 3D coverage strategy. The environmental parameters are listed in Table \ref{tab:PLOSParameter}\cite{bor2016efficient}. In simulations, the communication and on-board circuit power are normalized by noise power $N_0=-174$ dbm/hz. The carrier frequency $f_{\rm{c}}=2.4$ GHz. Without loss of generality, let ${A(\beta)}/{(\pi E_{\rm{b}})}=1~\rm{m}^2/\rm{J}$ and assume the desired data rate to be 1 bit/s/Hz.

\begin{table}[tbp]
 \caption{LOS probability parameters}\label{tab:PLOSParameter}
 \centering
 \begin{tabular}{|c|c|}
  \hline
  \textbf{Environment}&  \textbf{Parameters} $(a,b,\eta_{0},\eta_{1})$ \\
  \hline  
  Suburban & (4.88,0.43,0.1,21) \\
  \hline
  Urban & (9.61,0.16,1,20) \\
  \hline
  Dense Urban & (12.08,0.11,1.6,23) \\
  \hline
  High-rise Urban & (27.23,0.08,2.3,34) \\
  \hline
 \end{tabular}
\end{table}

\begin{figure}[htbp]
\centering
\includegraphics[width=0.48\textwidth]{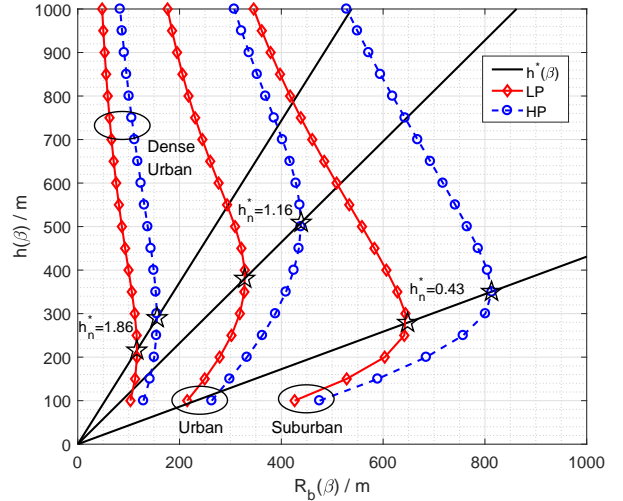}
\caption{The optimal hovering altitude versus $R_{\rm{b}}(\beta)$ in various environments. The red-solid line in dense urban, urban and suburban corresponds to transmit power 90 dB, 100 dB and 105 dB, respectively. Similarly, the blue-dash line corresponds to transmit power 95 dB, 105 dB and 110 dB, respectively.} \label{Fig:HversusRWithFixedPower}
\end{figure}

Fig. \ref{Fig:HversusRWithFixedPower} depicts the hovering altitude versus coverage range in various environments when the transmit power is fixed. In dense urban, urban and suburban, the red-solid lines correspond to transmit power with 90 dB, 100 dB and 105 dB, respectively. Similarly, the blue-dash line corresponds to 95 dB, 105 dB and 110 dB, respectively. The optimal hovering altitude is marked with stars. For example, observing the blue-dash line in suburban, one can find that with fixed transmit power, the coverage range achieves maximum when $h(\beta)=350$ m. In other words, with fixed coverage range, $h(\beta)=350$ m is the optimal hovering altitude that minimizes the transmit power. The solid black lines depict the optimal hovering altitude with respect to $R_{\rm{b}}(\beta)$. It can be seen that the optimal hovering altitudes in the same environment lie on a straight line. This is because the optimal hovering altitude is only determined by the desired coverage range and is proportional to the optimal hovering altitude when $R_{\rm{b}}(\beta)=1$ m, as previously illustrated in Corollary \ref{Cor:SameElevationAngle}. The slopes are labeled by $h^*_{\rm{n}}$ in Fig. \ref{Fig:HversusRWithFixedPower} and are determined by environment. Meanwhile, in high scattering environment, the optimal hovering altitude is also high. This is due to in high scattering environment, high hovering altitude decreases the average path loss by increasing LOS probability.

\begin{figure}[htbp]
\centering
\includegraphics[width=0.48\textwidth]{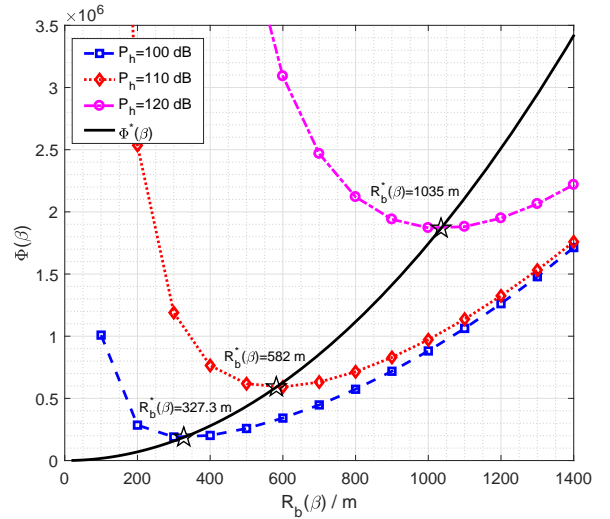}
\caption{The UAV-RF versus coverage range with various on-board circuit power. The optimal coverage ranges that minimize $\Phi(\beta)$ are marked by stars.} \label{Fig:PhiVersusR}
\end{figure}

The UAV-RF versus coverage range with various on-board circuit consumption power is depicted in Fig. \ref{Fig:PhiVersusR}. The simulation environment is urban, and the users' density is 0.1 /$\rm{m}^2$. The optimal coverages that minimize $\Phi(\beta)$ are marked by stars. When $P_{\rm{c}}$=100 dB, 110 dB and 120 dB, the simulated $R^*_{\rm{b}}(\beta)$=327.3 m, 582 m and 1035 m. As expected, high on-board circuit consumption power leads to high UAV-RF, which has been shown in Eq. \eqref{Equ:InequalityOfPhi}. This indicates that lowering the on-board circuit consumption power of UAV can decrease the UAV-RF effectively. Also, it can be observed that with the increase of on-board circuit consumption power, the optimal coverage range increases as well. This is because when $P_{\rm{c}}$ is high, large $R_{\rm{b}}(\beta)$ can decrease the number of UAV, resulting in the reduction of the total consumed on-board circuit power of network. The black-solid line depicts the optimal UAV-RF versus $R_{\rm{b}}(\beta)$. Substitute Eq. \eqref{Equ:Optimal2DArrangement} and Eq. \eqref{Equ:TransmitEqualToHovering} into Eq. \eqref{Equ:InequalityOfPhi}, the optimal UAV-RF can be expressed by
\begin{equation*}
\Phi^*(\beta) = \frac{2A(\beta)}{\pi E_{\rm{b}}}\lambda_{\rm{u}}(\beta,t)\left( 2^{S_{\rm{u}}} - 1 \right)R^{*2}_{\rm{b}}(\beta)\Gamma(h_{\rm{n}}^*(\beta)),
\end{equation*}
which agrees with the simulated results. 

\begin{figure}[htbp]
\centering
\includegraphics[width=0.48\textwidth]{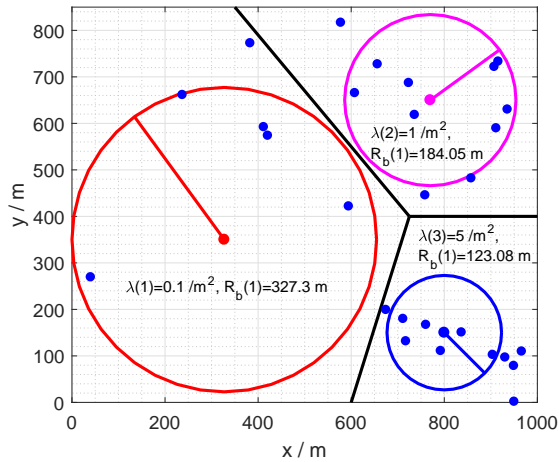}
\caption{The examples of UAV-BS 2D deployment with different users' density.} \label{Fig:CoverageExample}
\end{figure}

Fig. \ref{Fig:CoverageExample} shows an example of coverage system with UAV-BS. The simulated environment is urban, the on-board circuit power is 100 dB and the related users' densities in three subregions are 0.1 $/\rm{m}^2$, 1 $/\rm{m}^2$ and 5 $/\rm{m}^2$, respectively. The corresponding theoretical coverages range generated by Eq. \eqref{Equ:Optimal2DArrangement} are 327.3 m, 184.05 m and 123.08 m, which agree with the simulated results. It can be observed that when users' density is dense, the optimal coverage of each UAV-BS is small. This is due to small coverage range can reduce the transmit power increased by high users' density. Observe the three considered subregions, our proposed optimal 3D coverage strategy can be efficiently adjusted according to the varying of users' density.

\section{Conclusion}\label{Sec:Conclusion}

This paper focused on the downlink of UAV-BSs and proposed an optimal 3D placement that minimizes the UAV-RF, which is defined to characterize the life-time of network. The consumed power of on-board circuits including rotors, computational chips and gyroscopes, etc. are taken into account. By analyzing the optimal coverage of one single UAV, we first derived that the optimal hovering altitude is proportional to the coverage radius of UAVs, and the slope is only determined by communication environment. That is, dense scattering environment may greatly enlarge the needed hovering altitude. Then, by applying the derived optimal hovering altitude, the UAV-RF versus environment, coverage parameters and on-board circuit consumption power are derived. Simulation and theoretical results indicate that the minimum UAV-RF is achieved when the transmit power equals on-board circuit consumption power. That is, limiting on-board circuit power can effectively prolong the life-time of network. In addition, our proposed 3D placement method only requires the statistics of users' density and environment. As a typical on-line method, it can be easily implemented and can be utilized in scenarios with varying users' density.

\section*{Acknowledgment}

This work was partly supported by the China Major State Basic Research Development Program (973 Program) No. 2012CB316100(2) and National Natural Science Foundation of China (NSFC) No. 61321061.

\bibliography{bibfile}



\end{document}